\documentclass{article}

\usepackage[preprint]{neurips_2025}

\usepackage[utf8]{inputenc} 
\usepackage[T1]{fontenc}    
\usepackage{hyperref}       
\usepackage{url}            
\usepackage{booktabs}       
\usepackage{amsfonts}       
\usepackage{nicefrac}       
\usepackage{microtype}      

\usepackage{mathrsfs}
\usepackage{amsfonts}
\usepackage{graphicx}
\usepackage{mathrsfs}
\usepackage{amsmath}
\usepackage{algorithm}
\usepackage{algpseudocode}
\usepackage{amsthm} 
\usepackage{amssymb}
\usepackage{bbm}
\usepackage{float}
\usepackage{subfig}
\usepackage{rotating,color}
\usepackage{comment}
\usepackage{placeins}
\usepackage{xcolor,epstopdf}
\usepackage{appendix}

\newtheorem{theorem}{ \noindent T{\footnotesize HEOREM}}[section]

\newtheorem{lemma}{ \noindent L{\footnotesize EMMA}}[section]

\def\bm {\boldsymbol}

\def\bms{{\bf \Sigma}}
\def\S{{\bf S}}

\def\bmu{\bm \mu}
\def\X{\bm X}
\def\L{{\bf \Lambda}}
\def\Y{\bm Y}
\def\O{{\bf \Omega}}

\def\cov{{\rm Cov}}
\def\tr{{\rm tr}}

\title{Spatial Sign based Direct Sparse Linear Discriminant Analysis for High Dimensional Data}

\author{%
  \textbf{Dan Zhuang}$^{1}$~~~
    \textbf{Long Feng}$^{2}$\thanks{Corresponding author: Long Feng <flnankai@nankai.edu.cn>}\\
  $^{1}$School of Mathematics and Statistics,
  Fujian Normal University,\\
  Fuzhou, China \\
   $^{2}$School of Statistics and Data Science, KLMDASR, LEBPS, and LPMC, \\
   Nankai University, 
   Tianjin, China\\
   \texttt{zhgdan@fjnu.edu.cn, flnankai@nankai.edu.cn} \\
}

\begin{document}
\maketitle

\begin{abstract}
This paper investigates the robust linear discriminant analysis (LDA) problem with elliptical distributions in high-dimensional data. We propose a robust classification method, named SSLDA, that is intended to withstand heavy-tailed distributions. We demonstrate that SSLDA achieves an optimal convergence rate in terms of both misclassification rate and estimate error. Our theoretical results are further confirmed by extensive numerical experiments on both simulated and real datasets. Compared with current approaches, the SSLDA method offers superior improved finite sample performance and notable robustness against heavy-tailed distributions. The code is available at \href{https://github.com/RobustC/SSLDA}{https://github.com/RobustC/SSLDA}.
\end{abstract}

\section{Introduction}\label{sec: intro}

High-dimensional data is increasingly prevalent in various real-world applications, such as genomic data\citep{schafer2005shrinkage}, investment portfolio data\citep{ledoit2003improved}, and fMRl decoding\citep{shi2009sparse}. 
High-dimensional classification problems have garnered significant interest in recent decades. 
LDA is extensively utilized among numerous classification methods owing to its efficacy in practical applications. 
However, in high-dimensional settings, the classical LDA reveal the following critical limitations.
Firstly, the precision matrix is either non-estimable or exceedingly challenging to estimate because to the singularity and irreversibility of the covariance matrix; Secondly, there is a compounding of errors in the estimation of unknown parameters. \cite{bickel2004some} discovered that the efficacy of classical LDA may resemble that of random guessing in high-dimensional samples, despite the validity of the Gaussian assumption.

These limitations have resulted in the advancement of enhanced LDA techniques founded on certain sparse assumptions to address high-dimensional classification contexts. One typical strategy involves incorporating regularisation into the classification direction vector, such as the $\ell_1$ regularization \citep{witten2009covariance,mai2012direct,clemmensen2011sparse} or the $\ell_2$ regularization
\citep{guo2007regularized}, among others.
Another specific strategy involves assuming that the covariance matrix $\boldsymbol{\Sigma}$ and the mean difference $\boldsymbol{\delta} = \boldsymbol{\mu}_1 - \boldsymbol{\mu}_2$ are sparse, hence facilitating consistent estimation of them. In \cite{bickel2004some}, a naive Bayes rule or the independence principle is presented by substituting $\boldsymbol{\Sigma}$ with the diagonal of the sample covariance matrix.
\cite{shao2011sparse} proposed a sparse LDA method based on thresholding methodology and demonstrated the resultant can theoretically achieve the Bayes error. \cite{tibshirani2002diagnosis} and  \cite{fan2008high} proposed the nearest shrunken centroid estimation and the features annealed independent rule, wherein variable selection is executed through soft and hard thresholding criteria, respectively.

In contrast to the aforementioned approaches that separately estimation of $\boldsymbol{\Sigma}^{-1}$ and $\boldsymbol{\delta}$, some straightforward and efficient classifiers are introduced that directly estimate the product $\boldsymbol{\beta}^*=\boldsymbol{\Sigma}^{-1}\boldsymbol{\delta}$ by assumes the sparsity of the discriminant direction.
 \cite{cai2011direct} introduced the linear programming discriminant rule for sparse linear discriminant analysis based on direct estimation of $\boldsymbol{\beta}^*$ via limited $\ell_1$ minimization. 
The direct sparse discriminant analysis investigated in \cite{mai2012direct} represents another widely utilised sparse LDA approach, which is both computationally efficient and relatively straightforward to comprehend, as it reformulates the high-dimensional LDA into a penalised linear regression framework. In general, direct estimate $\boldsymbol{\beta}^*$ offers a considerable computational advantage over existing approaches that necessitate separate estimations of them; it is more economical and has strong performance, even when covariance or mean differences are not consistently assessed.

Note that LDA 
shows high sensitivity to outliers \citep{hastie2009elements}, particularly in high-dimensional datasets. Robust estimation of the mean and covariance is essential for the efficacy of LDA.
Optimized LDA performance occurs when normality and homoscedasticity are met \citep{croux2008classification}.
It is critical to point out that traditional robust mean estimation techniques, such as the coordinate-wise median and geometric median, experience degradation in high-dimensional spaces, as their error bounds are proportional to the dimensionality.
Numerous research investigate robust mean estimate for tainted data in high dimensions, such as Tukey’s median and Iterative Filtering\citep{diakonikolas2019robust,diakonikolas2017being}. Recent works have concentrated on reducing the spectral norm of weighted sample covariance and estimating the mean by a weighted average \citep{cheng2019high,zhu2022robust}. Motivated by this methodology, in \cite{deshmukh2023robust}, the $\ell_1$ norm of an outlier indicator vector is reduced subject to a constraint on the spectral norm of the weighted sample covariance, resulting in an order-optimal robust estimate of the mean. Based on $\ell_1$ norm, 
 \cite{li2019robust} introduces a robust discriminant analysis criterion, which is an upper bound of the theoretical framework of Bhattacharyya optimality.

 However, most existing methods
work well under (sub-) Gaussian assumptions but poorly on heavy-tailed distributions.
 For more general distributions, \cite{fang1990statistical} demonstrated that the Fisher rule remains best for elliptical distributions, which encompasses multivariate normal-$t$ and double exponential distributions. 
\cite{wakaki1994discriminant} explored Fisher's linear discriminant function for a wide class of elliptical symmetric distributions sharing a common covariance matrix.
To improve discriminant analysis method robustness and efficiency,
 \cite{andrews2011model} examined linear and quadratic discriminant classification with a mixture of multivariate-$t$ distributions.
 \cite{bose2015generalized} endeavored to extend linear and quadratic discriminant analysis to elliptically symmetric distributions.
Some studies have investigated the efficacy of classification and theoretical intricacies associated with elliptical distributions (e.g., \cite{cai2011direct,shao2011sparse,yang2023efficient}). Han et al. \cite{han2013coda} successfully generalized these principles to broader distributions utilizing the Gaussian linkage approach. Recently, by allowing each observation to originate from its own elliptically symmetric distribution,  \cite{houdouin2024femda} develops an innovative, robust, and model-free discriminant analysis algorithm.

Within the framework of elliptical distributions, employing a spatial-sign-based methodology has shown significant effectiveness, even in high-dimensional situations \citep{oja2010multivariate,raninen2021linear}.
Under high-dimensional elliptical populations, \cite{li2017spectral} investigates spatial-sign covariance matrix on its asymptotic spectral behaviors. \cite{wang2015high} introduced a nonparametric one-sample test utilizing the multivariate spatial sign transformation for elliptically distributed data. Based on spatial ranks and inner standardization, \cite{feng2016spatial} present spatial-sign-based test methodologies for high-dimensional one-sample localization problems. 
In the context of the two-sample location problem, \cite{feng2016multivariate} proposed a robust multivariate-sign-based procedures and  \cite{huang2023high} advocated for a two-sample inverse norm sign test. 
See \cite{zou2014multivariate}, \cite{feng2017high}, \cite{zhang2022robust}, \cite{feng2024spatial}, and the references therein for more related studies on spatial sign-based approaches.

In this work, we introduce a high-dimensional sparse LDA method designed to directly estimate the 'discriminant direction' under the assumption of elliptical distribution. The main contributions of the paper are summarized as below. 
\vspace{-2mm}
\begin{itemize}
    \item We establish theoretical results for Spatial-sign based linear discriminant analysis (SSLDA) in the sparse scenario.
    \item We show the consistency and rate of convergence results both misclassification rate and estimate error by assuming that the data is elliptically distributed.  
    \item We employ
       the sample spatial median and the spatial sign covariance matrix by directly estimate 'discriminant direction', which we demonstrate to be robust and efficient under fairly general assumptions. 
    \item Empirical studies are employed to evaluate image classification performance using datasets of braeburn apples and white cabbages. Results demonstrate that the proposed approach outperforms existing methods, showing promising performance in accuracy and efficiency.          
    \item Code is available at \href{https://github.com/RobustC/SSLDA}{https://github.com/RobustC/SSLDA}, supporting real and complex-valued data.
\end{itemize}
\vspace{-2mm}

The rest of the paper is organized as follows. A robust classification method, named SSLDA, is present in
Section~\ref{Sec:method}. In
Section~\ref{sec:Teor}, we investigate the asymptotic properties of SSLDA. 
Section~\ref{sec:simu} showcases the results of numerical simulations, while Section~\ref{sec:real} demonstrates a real-world application to image classification. Finally, discussion and limitation are present in Section~\ref{sec:conclu}. 

\emph{Notation:} For a vector ${\bm v}=(v_1, \ldots, v_p)^\prime$, we define 
$\|{\bm v}\|_0 = \sum_{j=1}^p\mathbbm{1}\{v_j\neq0\}, \|{\bm v}\|_1 = \sum_{j=1}^p |v_j|$ and $\|{\bm v}\|_2 = \sqrt{\sum_{j=1}^p v_j^2}$, $\|{\bm v}\|_\infty = \max_{1\leq j\leq p}|v_j|$, where $\mathbbm{1}(\cdot)$ is the indicator function that returns 1 if the condition inside the brackets is true, and 0 otherwise. 
For a matrix ${\bf M} = (a_{ij})_{p\times q }$, we define the matrix $\ell_1$ norm $\|{\bf M}\|_{L_1} = \max_{1 \leq j \leq q} \sum_{i=1}^p |a_{ij}|$, the elementwise $\ell_\infty$ norm 
$\|{\bf M}\|_\infty = \max\{|a_{ij}|\}$. $\lambda_{\min}({\bf M})$ and $\lambda_{\max}({\bf M})$ denote the smallest and largest eigenvalues of ${\bf M}$. 
$a_n \asymp b_n$ signifies $a_n=O(b_n)$ and $b_n=O(a_n)$ for any positive number sequences $\left\{a_n\right\}$ and $\left\{b_n\right\}$.
For any random variable $X \in \mathbb{R}$, we define the sub-Gaussian norm as $\|X\|_{\psi_2}:=\sup _{k \geq 1} k^{-1 / 2}\left(E|X|^k\right)^{1 / k}.$


\section{Robust Classification Method}\label{Sec:method}

In the context of classifying between two $p$-dimensional normal distributions, $\mathrm{N}\left(\boldsymbol{\mu}_1, \boldsymbol{\Sigma}\right)$ (designated as class 1) and $\mathrm{N}\left(\boldsymbol{\mu}_2, \boldsymbol{\Sigma}\right)$ (designated as class 2), both sharing the same covariance matrix $\boldsymbol{\Sigma}$, we consider a random vector $\mathbf{Z}$ that originates from one of these distributions with equal prior probabilities. The task is to determine the class to which $\mathbf{Z}$ belongs. When the parameters $\boldsymbol{\mu}_1$, $\boldsymbol{\mu}_2$, and $\boldsymbol{\Sigma}$ are known, Fisher's linear discriminant rule provides a straightforward solution. This rule is given by:
$$
\psi_{\mathbf{F}}(\mathbf{Z})=\mathbbm{1}\left\{(\mathbf{Z}-\boldsymbol{\mu})^{\prime} \boldsymbol{\Omega} \boldsymbol{\delta} \geq 0\right\},
$$
where $\boldsymbol{\mu} = \frac{\boldsymbol{\mu}_1 + \boldsymbol{\mu}_2}{2}$ is the midpoint between the two mean vectors, $\boldsymbol{\delta} = \boldsymbol{\mu}_1 - \boldsymbol{\mu}_2$ is the difference between the two mean vectors, $\boldsymbol{\Omega} = \boldsymbol{\Sigma}^{-1}$ is the inverse of the covariance matrix.

According to this rule, $\mathbf{Z}$ is classified into class 1 if $\psi_{\mathbf{F}}(\mathbf{Z})=1$, and into class 2 otherwise. Fisher's linear discriminant is the optimal classifier in this scenario, as it coincides with the Bayes rule when the prior probabilities for the two classes are equal.

In practice, we often don't know the true parameters, so we estimate them using samples. Suppose  $\left\{\mathbf{X}_k ; 1 \leq k \leq n_1\right\}$ and $\left\{\mathbf{Y}_k ; 1 \leq k \leq n_2\right\}$ are independent and identically distributed random samples from $\mathrm{N}\left(\boldsymbol{\mu}_1, \boldsymbol{\Sigma}\right)$ and $\mathrm{N}\left(\boldsymbol{\mu}_2, \boldsymbol{\Sigma}\right)$, respectively. Set $\hat{\bmu}_1$ and $\hat{\bmu}_2$ are the sample means of these two samples, respectively. And $\hat{\bms}=\frac{1}{n}\left(n_1\hat{\bms}_1+n_2\hat{\bms}_2\right), n=n_1+n_2$ where $\hat{\bms}_1,\hat{\bms}_2$ are the sample covariance matrix of these two samples, respectively.
So $\mathbf{Z}$ is classified into class 1 if
\begin{align*}
(\mathbf{Z}-\hat\bmu)^{\prime} \hat{\boldsymbol{\Omega}} \hat{\boldsymbol{\delta}} \geq 0
\end{align*}
where $\hat{\bmu}=\frac{\hat{\bmu}_1+\hat{\bmu}_2}{2}, \hat{\bm \delta}=\hat{\bmu}_1-\hat{\bmu}_2$ and $\hat{\O}=\hat{\bms}^{-1}$. When the dimension $p$ is larger than the sample sizes, $\hat{\bms}$ is not invertible. So the above classification rule can not work well. Consequently, many literatures consider the high dimensional linear discriminant analysis, such as \cite{cai2011direct,le2020adapted,park2022high}. In an important work, 
\cite{cai2011direct} proposed a direct estimation approach to sparse linear discriminant analysis. They estimate $\bm \beta^*=\O\bm \delta$ directly by the solution to the following optimization problem:
\begin{align}\label{cl1}
\hat{\boldsymbol{\beta}} \in \underset{\boldsymbol{\beta} \in \mathbb{R}^p}{\arg \min }\left\{\|\boldsymbol{\beta}\|_1 \text { subject to }\left\|\hat{\boldsymbol{\Sigma}} \boldsymbol{\beta}-(\hat{\bmu}_1-\hat{\bmu}_2)\right\|_{\infty} \leq \lambda_n\right\},
\end{align}
where  $\lambda_n$  is a tuning parameter. The constrained $\ell_1$ minimization method (\ref{cl1}) is known to be an effective way for reconstructing sparse signals, see 
\cite{donoho2005stable} and \cite{candes2007dantzig}.
Then, they proposed a new classification rule: $\mathbf{Z}$ is classified into class 1 if
\begin{align*}
(\mathbf{Z}-\hat\bmu)^{\prime} \hat{\bm \beta} \geq 0.
\end{align*}

However, the above methods are all constructed based on the sample means and covariance matrix which do not perform very well for heavy-tailed distributions. In this paper, we assume $\X$ and $\Y$ are generated from the elliptical distribution with density
\begin{align*}
|\L|^{-1 / 2} g\left((x-\bmu_k)^{\prime} \L^{-1}(x-\bmu_k)\right), k=1,2,
\end{align*}
respectively, where $g(\cdot)$ is a decreasing function. Without loss of generality, we assume that $\tr(\L)=p$. If the covariance matrix $\bms=\cov(\X)=\cov(\Y)$ exist, $\bms=\omega \L$ with positive parameter $\omega=p^{-1}\tr(\bms)\in\mathbb{R}$. \cite{fang1990statistical} showed that Fisher’s rule is still optimal for elliptical distributions. Since the constant \(\omega\) does not affect the decision rule, we can, without loss of generality, set \(\omega = 1\). In this case, we use the sample spatial median and spatial-sign covariance matrix to replace the sample mean and covariance matrix.

We often use the spatial median to estimate $\bm \mu$, i.e.
\begin{align}\label{emu}
\tilde{\bm \mu}_1=\underset{\bm \mu \in \mathbb{R}^p}{\arg\min} \sum_{i=1}^{n_1}\|\bm X_i-\bm \mu\|_2, ~\tilde{\bm \mu}_2=\underset{\bm \mu \in \mathbb{R}^p}{\arg\min} \sum_{i=1}^{n_2}\|\bm Y_i-\bm \mu\|_2
\end{align}
Let $U(\bm x)=\frac{\bm x}{\|\bm x\|_2}I(\bm x\not=\bm 0)$. Then the sample spatial sign covariance matrix is defined as
\begin{align}\label{hats}
\hat{{\bf S}}_1=\frac{1}{n_1}\sum_{i=1}^{n_1} U(\bm X_i-\tilde{\bm \mu}_1)U(\bm X_i-\tilde{\bm \mu}_1)^\prime, ~\hat{{\bf S}}_2=\frac{1}{n_2}\sum_{i=1}^{n_2} U(\bm Y_i-\tilde{\bm \mu}_2)U(\bm Y_i-\tilde{\bm \mu}_2)^\prime
\end{align}
and the population spatial-sign covariance matrix is estimated by $\hat{\S}=\frac{1}{n}\left(n_1\hat{\S}_1+n_2\hat{\S}_2\right)$. We estimate $\bm \gamma^*=\L^{-1}\bm \delta$ directly by the solution to the following optimization problem:
\begin{align}\label{cl}
\hat{\boldsymbol{\gamma}} \in \underset{\boldsymbol{\gamma} \in \mathbb{R}^p}{\arg \min }\left\{\|\boldsymbol{\gamma}\|_1 \text { subject to }\left\|p\hat{\S} \boldsymbol{\gamma}-(\tilde{\bmu}_1-\tilde{\bmu}_2)\right\|_{\infty} \leq \lambda_n\right\},
\end{align}
and the corresponding classification rule is
$\mathbf{Z}$ is classified into class 1 if
\begin{align}\label{scl1}
(\mathbf{Z}-\tilde\bmu)^{\prime} \hat{\bm \gamma} \geq 0.
\end{align}
where $\tilde\bmu=\frac{\tilde\bmu_1+\tilde\bmu_2}{2}$.

\section{Theoretical Results} \label{sec:Teor}
The optimal misclassification rate in this case is
$$
R: =\frac{1}{2} \mathrm{P}\left((\X-\bmu_1)^{\prime} \boldsymbol{\Omega} \boldsymbol{\delta}<-\frac{1}{2} \boldsymbol{\delta}^{\prime} \boldsymbol{\Omega} \boldsymbol{\delta}\right)+\frac{1}{2} \mathrm{P}\left((\Y-\bmu_2)^{\prime} \boldsymbol{\Omega} \boldsymbol{\delta} \geq \frac{1}{2} \boldsymbol{\delta}^{\prime} \boldsymbol{\Omega} \boldsymbol{\delta}\right) .
$$
As in the work of 
\cite{shao2011sparse}, under the assumption of elliptical distribution, for any $p$-dimensional non-random vector ${\bm u}$ with $\|\bm u\|_2=1$ and any $t \in \mathbb{R}$,
$$
\mathrm{P}\left({\bm u}^{\prime} \boldsymbol{\Omega}^{1 / 2} (\X-\bmu_1) \leq t\right)=: \Psi(t)
$$
is a continuous distribution function symmetric about 0 and does not depend on ${\bm u}$. Given $\left\{\mathbf{X}_k\right\}$ and $\left\{\mathbf{Y}_k\right\}$, the conditional classification error of the linear programming discriminant (LPD) rule 
 is
$$
R_n:=1-\frac{1}{2} \Psi\left(-\frac{\left(\tilde{\boldsymbol{\mu}}-\boldsymbol{\mu}_1\right)^{\prime} \hat{\boldsymbol{\gamma}}}{\left(\hat{\boldsymbol{\gamma}}^{\prime} \boldsymbol{\Sigma} \hat{\boldsymbol{\gamma}}\right)^{1 / 2}}\right)-\frac{1}{2} \Psi\left(\frac{\left(\tilde{\boldsymbol{\mu}}-\boldsymbol{\mu}_2\right)^{\prime} \hat{\boldsymbol{\gamma}}}{\left(\hat{\boldsymbol{\gamma}}^{\prime} \boldsymbol{\Sigma} \hat{\boldsymbol{\gamma}}\right)^{1 / 2}}\right) .
$$
where $\hat{\boldsymbol{\gamma}}$ is given in (\ref{cl}). The efficacy of the LPD rule can be effectively gauged through the difference (or ratio) between $R_n$ and $R$. Let $\Delta_p=\boldsymbol{\delta}^{\prime} \boldsymbol{\Omega} \boldsymbol{\delta}=\omega^{-1}\bm \delta^\prime \L^{-1}\bm \delta$. $\{\sigma_{ii}=\mathbf{\Sigma}_{ii}\}_{i=1}^p$ denote the corresponding marginal variances. Refer to \cite{cai2011direct}, we need provide the following conditions before presenting the difference and ratio between $R_n$ and $R$. 
\begin{itemize}
\item[(C1)] $n_1 \asymp n_2, \log p \leq n, c_0^{-1}\le \lambda_{\min}(\mathbf{\Sigma})\le \lambda_{\max}(\mathbf{\Sigma})\le c_0, \max\limits _{1 \leq i \leq p} \sigma_{i i} \leq K$ and $\Delta_p \geq c_1$ for some constant $K>0$ and $c_0,c_1>0$.
\item[(C2)] Define $\zeta_k=\mathbb{E}\left(\xi_i^{-k}\right),  \xi_i=\left\|\boldsymbol{X}_i-\boldsymbol{\mu}\right\|_2,  \nu_i=\zeta_1^{-1} \xi_i^{-1}$.
(1) $\zeta_k \zeta_1^{-k}<\zeta \in(0, \infty)$ for $k=1,2,3,4$ and all $p$.
(2) $\lim \sup _p \lambda_{\max}(\mathbf{S})<1-\psi<1$ for some positive constant $\psi$.
(3) $\nu_i$ is sub-gaussian distributed, i.e. $\left\|\nu_i\right\|_{\psi_2} \leq K_\nu<\infty$.
\end{itemize}
Condition (C1) is the same as condition (C1) in \cite{cai2011direct}, which is commonly used conditions in the high dimensional setting. Condition (C2) are consistent with conditions (A1-A2) in \cite{feng2024spatial}, which ensure the consistency of the spatial median estimator (\ref{emu}).

Then, we having the following theoretical results.

\begin{theorem}\label{th1}
Let $\lambda_n=C \sqrt{\Delta_p \log p / n}$ with $C$ being a sufficiently large constant. Suppose (C1)-(C2) hold and
\begin{align}\label{the1}
\frac{\|\boldsymbol{\L^{-1}} \boldsymbol{\delta}\|_1}{\Delta_p^{1 / 2}}+\frac{\|\boldsymbol{\L^{-1}} \boldsymbol{\delta}\|_1^2}{\Delta_p^2}=o\left(\sqrt{\frac{n}{\log p}}\right).
\end{align} 
Then we have as $n \rightarrow \infty$ and $p \rightarrow \infty$,
$$
R_n-R \rightarrow 0
$$
in probability.
\end{theorem}

\begin{theorem}\label{th2}
Let $\lambda_n=C \sqrt{\Delta_p \log p / n}$ with $C$ being a sufficiently large constant and $\frac{n}{p\log p} \to 0$. Suppose (C1)-(C2) hold and
$$
\|\boldsymbol{\L^{-1}} \boldsymbol{\delta}\|_1 \Delta_p^{1 / 2}+\|\boldsymbol{\L^{-1}} \boldsymbol{\delta}\|_1^2=o\left(\sqrt{\frac{n}{\log p}}\right).
$$
Then
$$
\frac{R_n}{R}-1=O\left(\left(\|\boldsymbol{\L^{-1}} \boldsymbol{\delta}\|_1 \Delta_p^{1 / 2}+\|\boldsymbol{\L^{-1}} \boldsymbol{\delta}\|_1^2\right) \sqrt{\frac{\log p}{n}}\right)
$$
with probability greater than $1-O\left(p^{-1}\right)$. In particular, if (C1)-(C2) hold and
$$
\|\boldsymbol{\L^{-1}} \boldsymbol{\delta}\|_0 \Delta_p=o\left(\sqrt{\frac{n}{\log p}}\right),
$$
then
$$
\frac{R_n}{R}-1=O\left(\|\boldsymbol{\L^{-1}} \delta\|_0 \Delta_p \sqrt{\frac{\log p}{n}}\right)
$$
with probability greater than $1-O\left(p^{-1}\right)$.
\end{theorem}

Theorem \ref{th1} and \ref{th2} is similar to the results in Theorem 2 and 4 in \cite{cai2011direct}. Theorem \ref{th1} show the consistency of our proposed method and Theorem \ref{th2} establish the rate of convergence.

\section{Simulation Studies} \label{sec:simu}
In this section, we investigate the empirical performance of the SSLDA method. 

\subsection{Implementation of SSLDA}

The estimate of $\bm \gamma^*$ is obtained by solving $\ell_1$ minimization problem of (\ref{cl}). This convex optimazation problem can be reformulated as the following linear program
$$\min \sum\limits_{j=1}^p u_j$$
\begin{align}\label{LP}
	\text{subject to:}&-\gamma_j\leq u_j~~~\text{for all} 1\leq j \leq p,\nonumber\\
	& +\gamma_j\leq u_j~~~\text{for all} 1\leq j \leq p,\\
	& -p\hat{\bm \sigma}_k^\prime \hat{\bm \gamma}_k + \tilde{\bm \delta}_k\leq \lambda_n~~~\text{for all} 1\leq k\leq p, \nonumber\\
	& +p\hat{\bm \sigma}_k^\prime \hat{\bm \gamma}_k+ \tilde{\bm \delta}_k\leq \lambda_n~~~\text{for all} 1\leq k\leq p,\nonumber
\end{align}
where $(\tilde{\delta}_1,\cdots,\tilde{\delta}_p):=\tilde{\bm \delta}$ and $(\hat{\bm \sigma}_1,\cdots,\hat{\bm \sigma}_p):=\hat{\S}$. We applied the CLIME method \citep{cai2011constrained} to solve (\ref{LP}). We replace $\hat{\S}$ in (\ref{cl}) by $\hat{\S}_\rho = \hat{\S}+\rho \bm I_{p\times p}$ with a small positive number $\rho$ (e.g., $\rho=\sqrt{\log p/n}$). 

The algorithm's tuning parameter $\lambda=\lambda_n$ can be optimized through empirical 10-fold cross-validation (CV). To implement this, partition the sets $\{1,2,\cdots,n_1\}$ and $\{1,2,\cdots,n_2\}$ into $2K$ subgroups $G_{ik}$, where $i=1,2, k=1,2,\cdots, K$. This division naturally separates the sample data $\left\{\mathbf{X}_i ; 1 \leq i \leq n_1\right\}$ and $\left\{\mathbf{Y}_j ; 1 \leq j \leq n_2\right\}$ into $K$ validation subsets $ \mathscr{X}_k:=\{\mathbf{X}_i, \mathbf{Y}_j: i\in G_{1k},j\in G_{2k}\}, 1\leq k\leq K$. Denote $\tilde{\bm\mu}_{(k)}$, $\hat{\bm S}_{(k)}$ be defined in (\ref{emu}) and (\ref{hats}) derived from $\left\{\mathbf{X}_k, \mathbf{Y}_k; 1 \leq k \leq n_1\right\}$\textbackslash$\mathscr{X}_k$. Based on $\tilde{\bm\mu}_{(k)}$, $\hat{\bm S}_{(k)}$, we can obtain $\hat{\boldsymbol{\gamma}}$ for a given $\lambda_n$ by (\ref{cl}). The final selection of $\lambda$ boils down to
$$
\hat{\lambda}=\max_\lambda \sum\limits_{k=1}^{K}\left(\sum\limits_{i\in G_{1k}}I_{i1}^{(k)}+\sum\limits_{j\in G_{2k}}I_{j2}^{(k)}\right),
$$
where $I_{j1}^{(k)}=1$ if $\mathbf{X}_i \in \mathscr{X}_k$ stisfies (\ref{scl1}), else $I_{j1}^{(k)}=0$; let $I_{j2}^{(k)}=1$ if $\mathbf{Y}_j\in \mathscr{X}_k$ not stisfies (\ref{scl1}), else $I_{j2}^{(k)}=0$, and we choose $K=10$ in this paper.

\subsection{Simulation Results}

We compare the numerical performance of the SSLDA method with fellowing methods:

\begin{itemize}
\item LS-LDA: the least square formulation for classification proposed by \cite{mai2012direct}.

\item CODA: Copula discriminant analysis classifier (CODA) for high-dimensional data proposed by \cite{han2013coda}.

\item LDA-CLIME: Linear programming discriminant for high-dimensional data clssification using CLIME method proposed by \cite{cai2011constrained}.


\item SSLDA: Sparse spatial-sign based linear discriminant analysis.
	
\end{itemize}

In the simulation studies, we fix the sample sizes $n=400$
and varied $p$ to be $\{100,200,300\}$.
Let $\boldsymbol{\mu}_1=\boldsymbol{0}, \boldsymbol{\mu}_2=(1,\cdots,1,0,\cdots,0)$, where the number of 1's is $s_0=10$. We generate $p$ dimensional predictors $x$ from the following four different elliptical distributions:

\begin{itemize}
	\item[(\uppercase\expandafter{\romannumeral1})]  Multivariate normal distribution: $\mathbf{Z}\sim N\left(\boldsymbol{\mu}_1, \boldsymbol{\Sigma}\right)$.
	
	\item[(\uppercase\expandafter{\romannumeral2})]  Multivariate $t_2$-distribution, data are generated from standardized $t_2/\sqrt{2}$ with mean $\boldsymbol{\mu}_1$ and $\boldsymbol{\Sigma}$.
	
	\item[(\uppercase\expandafter{\romannumeral3})]  Standardized multivariate mixture normal distribution $MN_{N,\kappa,9}=[\kappa N(\boldsymbol{0},\boldsymbol{\Sigma})+(1-\kappa)N(\boldsymbol{0},9\boldsymbol{\Sigma})]/\sqrt{\kappa+9(1-\kappa)}$. $\kappa$ is chosen to be 0.8.

 	\item[(\uppercase\expandafter{\romannumeral4})]  Cauchy distribution.
	
\end{itemize}

Based on the above four distributions, we considered the following two models.

\begin{itemize}
	\item \textbf{Model 1}: $\boldsymbol{\Sigma}_{i,j}=1-0.5\times\mathbbm{1}\{|i-j|\neq0\}, i,j=1,\cdots,p$. 
	
	\item \textbf{Model 2}: $\boldsymbol{\Sigma}_{i,j}=0.8^{|i-j|}, i,j=1,\cdots,p$.
	
	
\end{itemize}



In line with the simulation parameters outlined in \cite{cai2011direct}, this study fix $n_1=n_2=200$.
The average classification errors for the test samples and the standard deviations based on 100 replications are stated in Table \ref{tab1} and Table \ref{tab2}. Table \ref{tab1} and Table \ref{tab2} show the performance of SSLDA and three state-of-the-art methods across four distinct distributions of varying dimensionality. These four distributions serve to characterise the level of noise present in the original data distribution. The results show that SSLDA outperforms other methods across both Model 1 and Model 2. As the original distribution's tail thickens, SSLDA's benefits grow significantly, showcasing notable robustness, especially in Cauchy distribution.

\makeatletter\def\@captype{table}\makeatother
\begin{table}
	\caption{The average classification error and the standard deviation (in brackets) for the test samples in percentage for Model 1 over 100 Monte Carlo replications. (\%)}
	\label{tab1}
	\centering
 	\begin{tabular}{cccccc}
        \toprule
		Distribution&$p$ & SSLDA & LDA-CLIME &CODA & LS-LDA \\
		\hline  
		&100 &2.82(0.83) &3.20(1.04)  &2.75(0.83)   &2.98(0.93)  \\ 
		Normal&200 &3.23(0.97) &4.75(1.64)   &2.68(0.86)  & 2.71(0.82) \\
		&300&4.13(1.11) &3.84(1.10)    &2.34(0.77)& 2.57(0.83)\\
		\hline
		&100 &10.03(1.80) &11.89(1.94)& 10.84(1.49)& 11.26(1.90) \\ 
		$t_2$&200 & 10.60(1.57)  &12.24(1.73)  &10.40(1.72)&10.62(1.66)   \\
		&300 &11.47(1.79) &13.47(2.09)&10.64(1.91) &10.75(1.94) \\
		\hline  
	  &100 &7.58(1.28)&8.28(1.32)  &  8.10(1.29) &8.32(1.44) \\ 
	  $MN_{N,\kappa,9}$&200 & 7.22(1.40)& 8.51(1.36) &7.85(1.42) &8.15(1.48)\\
	  &300 & 7.84(1.43)& 8.65(1.43)   &7.59(1.25)&7.78(1.19)\\
		\hline  
	 &100 &15.64(1.97)& 20.31(2.28) &25.81(12.71) &20.60(8.43)\\ 
	 Cauchy&200 &16.33(2.26)& 20.49(2.61)&24.28(10.47)&18.19(2.52)\\
	 &300 &16.90(2.07)& 22.58(3.14)   &25.45(12.15)&17.53(2.40)\\
	\bottomrule
	\end{tabular}
\end{table}

\makeatletter\def\@captype{table}\makeatother
\begin{table}
	\caption{The average classification error and the standard deviation (in brackets) for the test samples in percentage for Model 2 over 100 Monte Carlo replications. (\%)}
	\label{tab2}
	\centering
 	\begin{tabular}{cccccc}
        \toprule
		Distribution&$p$ & SSLDA & LDA-CLIME &CODA & LS-LDA \\
		\hline  
		&100  &18.83(2.00) & 18.42(2.10)  & 18.08(2.14) & 17.93(1.97)\\
		Normal&200 &19.93(2.27) & 19.66(2.36)  & 18.01(2.16) & 18.51(2.27) \\ 
		&300 & 20.54(2.29)&20.34(2.12)& 18.59(1.77) &18.50(2.01)\\
		\hline  
	  &100  &23.67(2.17)  & 26.99(2.77) & 
        24.60(2.29) &26.28(2.81) \\
	  $t_2$&200 &24.91(2.19)    & 29.03(2.66)  
        &25.73(2.63)&26.82(2.96) \\
	  &300 &25.39(2.38)   & 28.78(2.53) 
        &25.79(2.96)&27.12(3.25) \\ 
		\hline
	  &100  & 22.78(2.38)&24.74(2.26)&23.42(2.42) &  
        23.96(2.29) \\
	  $MN_{N,\kappa,9}$&200 &22.83(2.21)& 25.62(2.51)&24.04(2.49) & 
        24.36(2.31) \\
	  &300 &23.99(2.54)& 25.68(2.90)&24.18(2.55)  
        &25.07(2.84)\\ 
		\hline  
	  &100  &27.40(1.95)&35.08(2.86)  & 37.78(7.96)& 
        33.93(3.55)  \\
	  Cauchy&200 &28.01(2.48) & 35.42(3.38) & 38.18(7.51)& 
        35.62(9.62) \\
	  &300 & 29.35(2.37) &35.74(3.24)& 38.47(6.81) 
        &35.58(7.18) \\ 
	    \bottomrule
	\end{tabular}
\end{table}

To further investigate the impact of truly influential features number on classification performance, we fixed $n_1=n_2=200$ and $p=100$, and examined the classification error rates of four distinct methods across different $s_0$ values.
We run the four methods on each $s_0\in [5,80]$, each repeated for 100 times. The averaged misclassification errors in percentage versus various $s_0$ are illustrated in Figure \ref{fig11}. It can be observed in Figure \ref{fig11} that SSLDA performs the best (blue curve) in different distributions, especially in heavy-tailed distributions (such as Cauchy distribution). The experimental results further highlight the superiority of the SSLDA method under high-dimensional heavy-tailed settings,  demonstrating its versatility in handling both sparse and dense mean vector configurations with equal effectiveness.

\begin{figure}
    \centering
    \subfloat{\includegraphics[width = 0.45\textwidth]{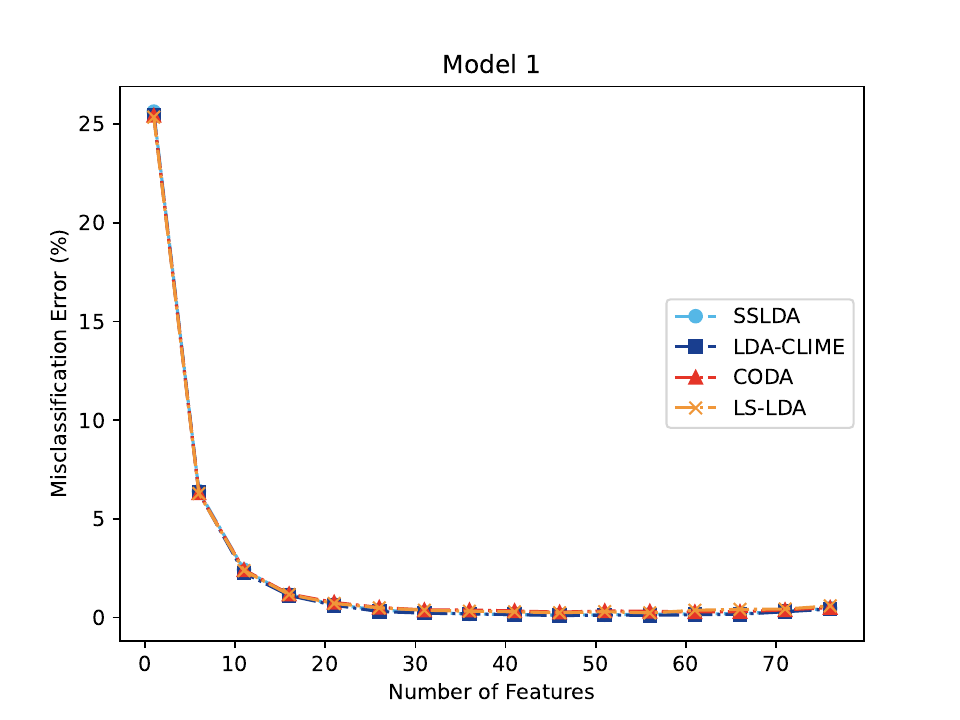}}
     \subfloat{\includegraphics[width = 0.45\textwidth]{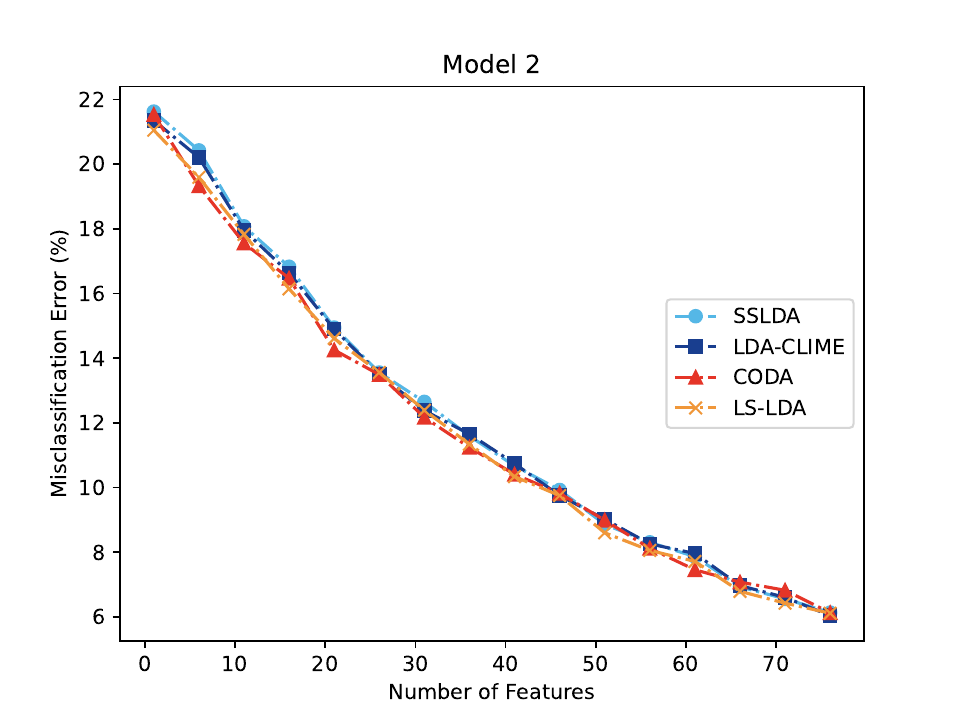}}\\
    (a) Normal distribution\\
    \subfloat{\includegraphics[width = 0.45\textwidth]{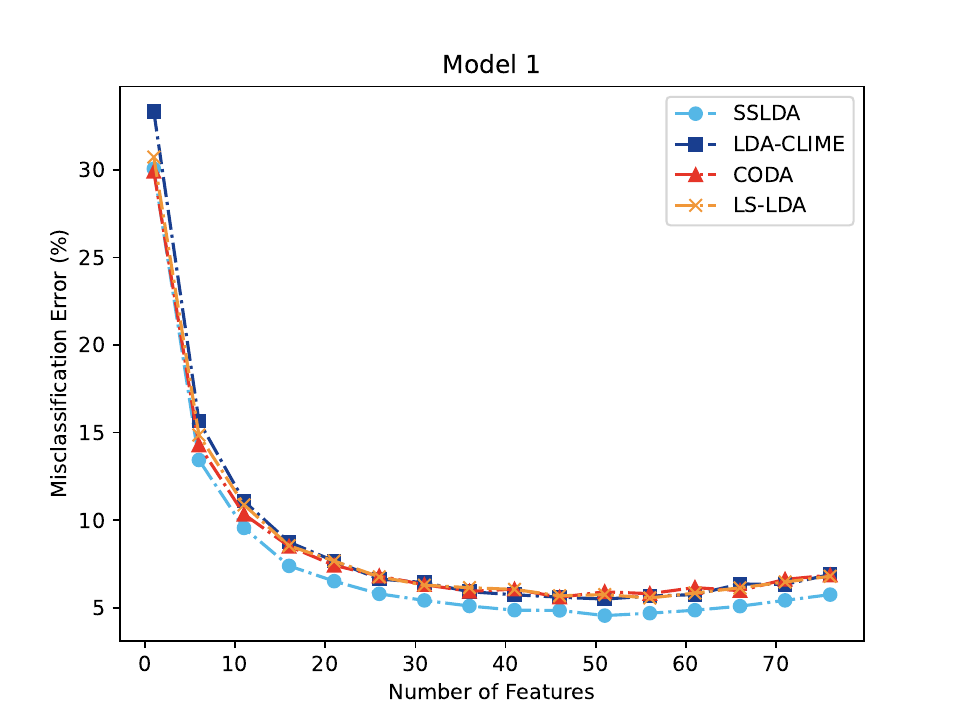}}
     \subfloat{\includegraphics[width = 0.45\textwidth]{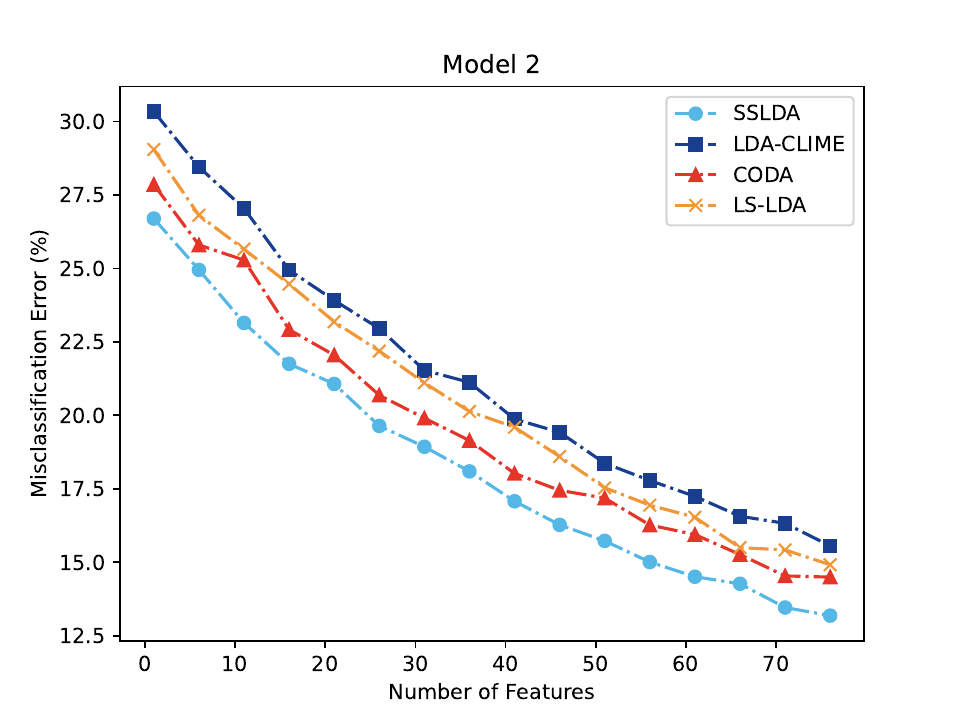}}\\
    (b) $t_2$ distribution\\
        \subfloat{\includegraphics[width = 0.45\textwidth]{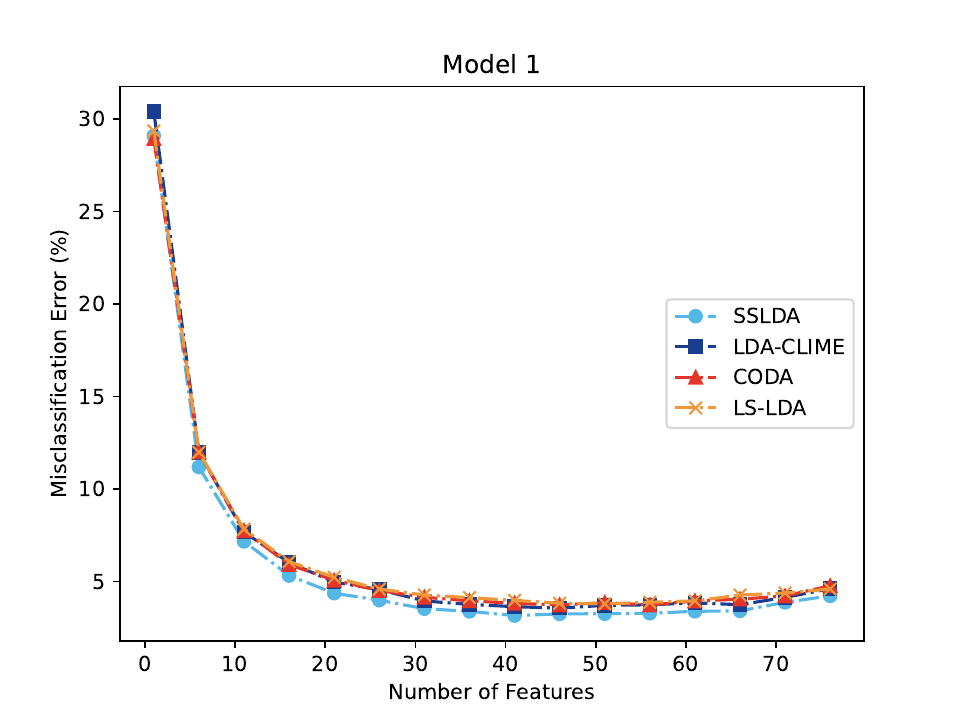}}
     \subfloat{\includegraphics[width = 0.45\textwidth]{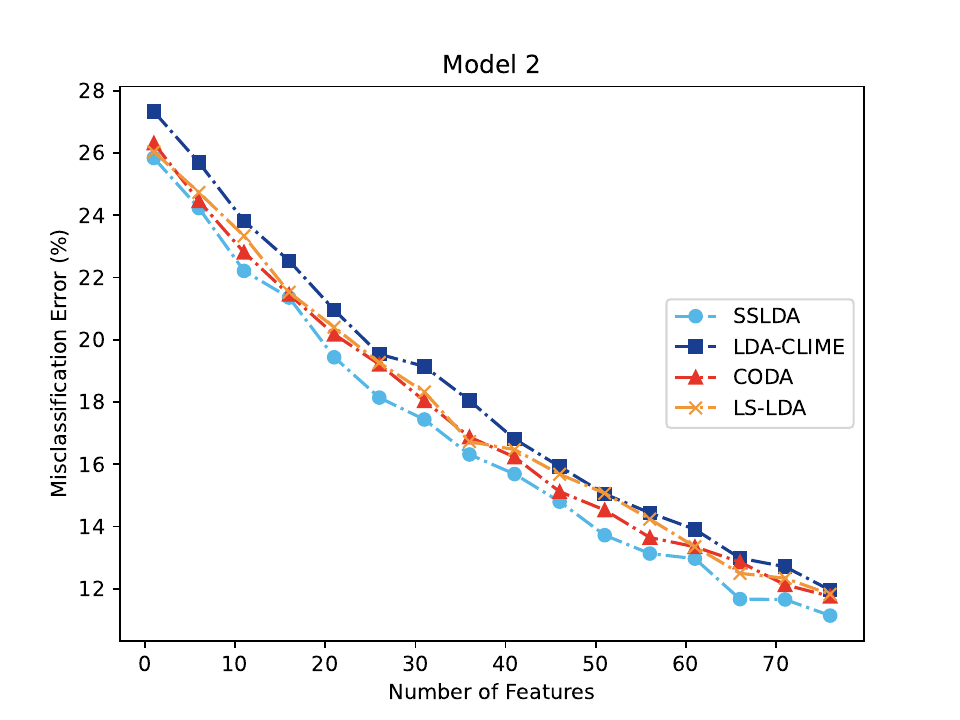}}\\
    (c) $MN_{N,\kappa,9}$ distribution\\
    \subfloat{\includegraphics[width = 0.45\textwidth]{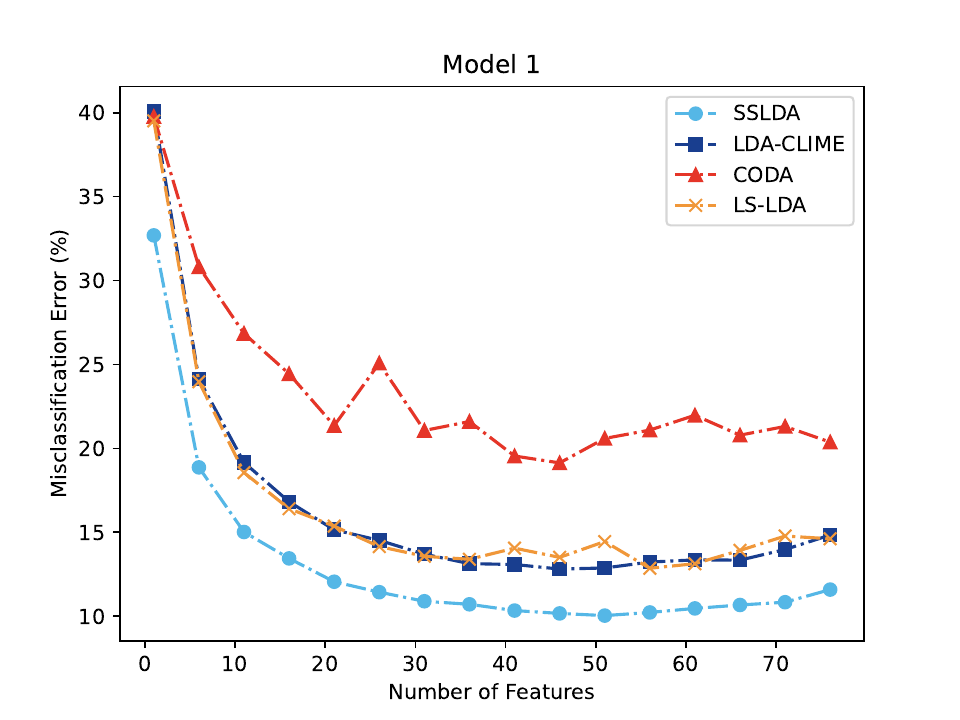}}
     \subfloat{\includegraphics[width = 0.45\textwidth]{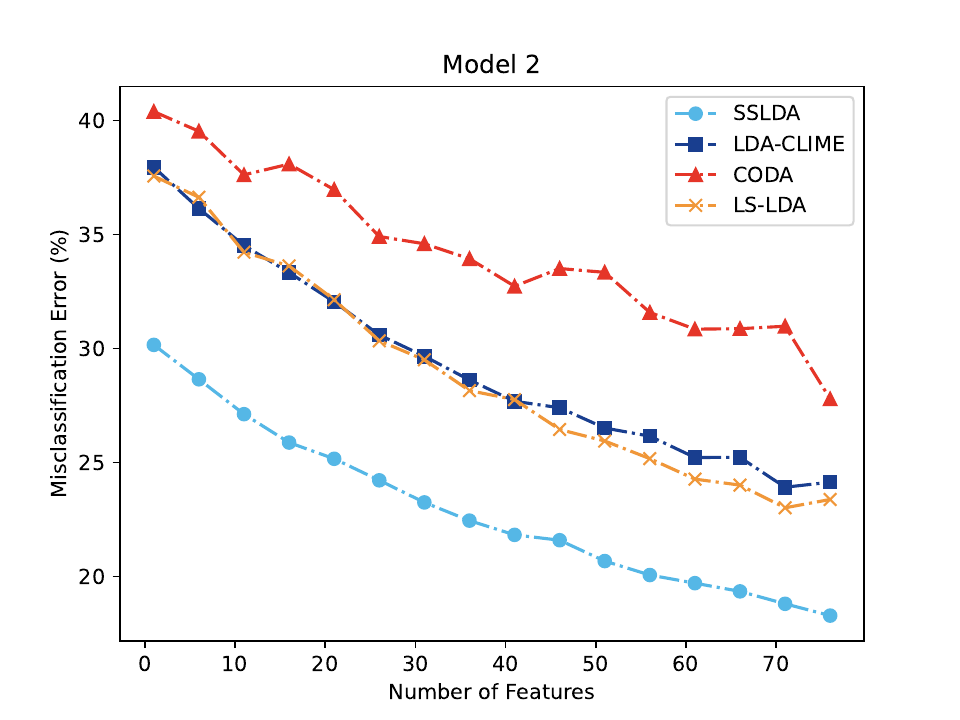}}\\
    (d) Cauchy distribution 
    \caption{The average misclassification error curves vary in different numbers of features $s_0$ on different distributions with four different methods over 100 Monte Carlo replications ($n_1=n_2=200,p=100$).}
    
    \label{fig11} 
\end{figure}
\FloatBarrier

\section{Real Data Application} \label{sec:real}
In this section, we apply the SSLDA classifier to the analysis of the real dataset to further examine the performance of the proposed rule. 


\textbf{Datasets.} There are 636 JPG images of two groups, braeburn apple and white cabbage, with quantities
of 492 and 144, respectively. All images are 100×100 pixels in size. This image set is sourced
from \href{https://www.kaggle.com/datasets/moltean/fruits}{https://www.kaggle.com/datasets/moltean/fruits}. The task of this dataset is to perform image classification on these images of Braeburn apples and white cabbages, to correctly distinguish between apples and cabbages. Figure \ref{figReal1} shows example images from each group along with the dataset's preprocessing stages. Initial real data refinement involved the steps below.

{\itshape Step 1.} To prevent the comparison method from being affected by imbalanced data, we randomly selected 144 images ($n_1=n_2=144$) from the braeburn apples image dataset for the experiment, to maintain consistency with the number of white cabbages images. 

{\itshape Step 2.} Each RGB image was converted to grayscale (black and white) using a assigned RGB ratio.
Subsequently, an ORB (Oriented FAST and Rotated BRIEF) keypoint detector \citep{rublee2011orb,daradkeh2021methods} was used to determine the descriptors of keypoints for each grayscale image, and then calculate the column mean of the descriptors for simple dimensionality reduction ($p=32$).
This process is implemented in JupyterLab 4.0.11 using the OpenCV library and the Python 3.6 programming language.

Following the preliminary processing stages of {\itshape Step 1.} and {\itshape Step 2.}, we obtain a final data matrix ($144\times 32$) for each group to distinguish braeburn apple from white cabbage. To evaluate method performance, we randomly split the data into equal training and testing sets (each 72 samples, i.e. $72\times 32$ matrix) without replacement. Each of the four methods is applied to the training set and assessed on the testing set, repeated 100 times.

\textbf{Results.} To evaluate algorithm efficacy in the classification task, the following metrics are employed:
\begin{align*} 
\text{Specificity} &= \frac{\text{TN}}{\text{TN}+\text{FP}},~~~
\text{Sensitivity} = \frac{\text{TP}}{\text{TP}+\text{FN}},\\
\text{Precision} &= \frac{\text{TP}}{\text{TP}+\text{FP}},~~~
\text{Accuracy} = \frac{\text{TP}+\text{TN}}{\text{TP}+\text{TN}+\text{FP}+\text{FN}},
\end{align*} 
where TP (True Positives) and TN (True Negatives) represent correct classifications for positive (class 1) and negative (class 2) cases, respectively. FP (False Positives) and FN (False Negatives), on the other hand, indicate misclassifications.
All these metrics fall within the range of 0 to 1. A higher score suggests that the classification algorithm is performing well, while a lower score indicates subpar performance. An algorithm hitting a perfect score of 1 across the board is essentially considered the gold standard.

Table \ref{tabreal1} shows the performance of different methods on the real-world dataset, which gives the four metrics of different methods. From Table \ref{tabreal1}, we can observe that SSLDA achieving the highest values for Specificity, Precision and Accuracy. Although SSLDA does not achieve the highest Sensitivity value, its performance is comparable to that of the other three methods. This results further confirms the reliability of SSLDA classifier method in high-dimensional scenario.

\begin{table}[ht]
	\caption{Comparisons of average (standard deviation) classification accuracy of Apple and Cabbage datasets over 100 replications.}
	\label{tabreal1}
	\centering
    \begin{tabular}{ccccc}
     \toprule
		Method & Specificity & Sensitivity&Precision & Accuracy \\
		\hline  
		SSLDA &0.9854(0.0131) & 0.9611(0.0219)&0.9852(0.0132)&0.9733(0.0130)\\ 
		LDA-CLIME & 0.9736(0.0255) &0.9639(0.0256) &0.9742(0.0240) &0.9688(0.0124) \\
		CODA&0.9541(0.0276)&0.9601(0.0272) &0.9550(0.0261) &0.9571(0.0200) \\
        LS-LDA&0.9686(0.0273)&0.9697(0.0211) &0.9694(0.0262)  &0.9691(0.0161)  \\
    \bottomrule
	\end{tabular}
\end{table}

\begin{figure}[htbp]
    \centering
    \subfloat{\includegraphics[width = 0.85\textwidth]{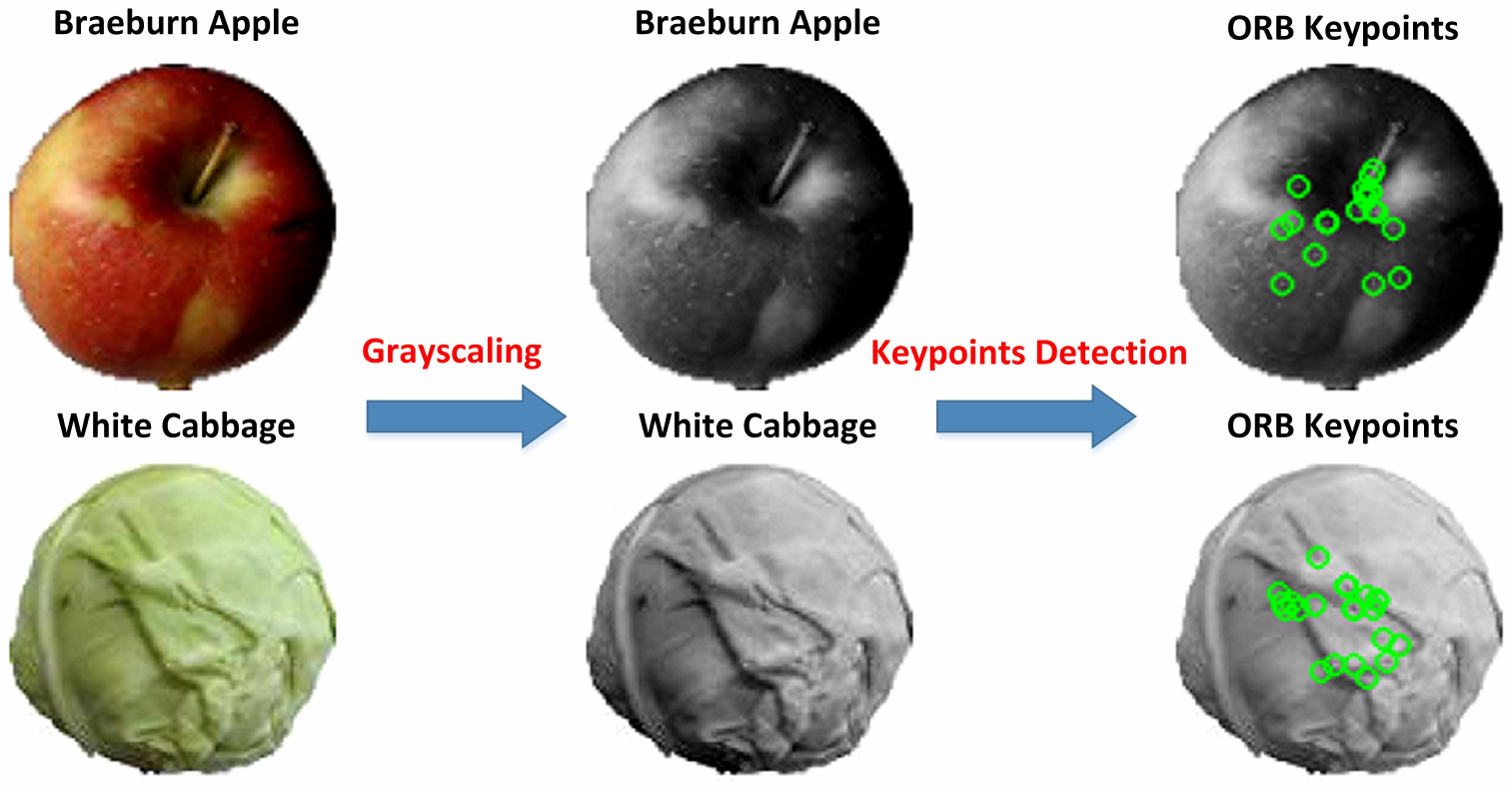}}    
    \caption{Preprocessing steps for the image dataset.}   
    \label{figReal1} 
\end{figure}
\FloatBarrier


\section{Discussion and Limitation}\label{sec:conclu}

Reliable classification for high-dimensional data, especially in the case of heavy-tail or untidy data, is typically a tricky for the applied statistician.
This paper proposed a robust classification approach that  capable of handling heavy-tailed data. The findings from the simulation trials showcase its superior enhanced finite sample prowess and notable computational efficiency both in synthetic data and real data. 
The primary limitation of this approach is its reliance on a stringent assumption of elliptical symmetric distribution.


{\small
	\bibliographystyle{plainnat}
	\bibliography{ref}

\begin{thebibliography}{46}
\providecommand{\natexlab}[1]{#1}
\providecommand{\url}[1]{\texttt{#1}}
\expandafter\ifx\csname urlstyle\endcsname\relax
  \providecommand{\doi}[1]{doi: #1}\else
  \providecommand{\doi}{doi: \begingroup \urlstyle{rm}\Url}\fi

\bibitem[Andrews et~al.(2011)Andrews, McNicholas, and Subedi]{andrews2011model}
Jeffrey~L Andrews, Paul~D McNicholas, and Sanjeena Subedi.
\newblock Model-based classification via mixtures of multivariate
  t-distributions.
\newblock \emph{Computational Statistics and Data Analysis}, 55\penalty0
  (1):\penalty0 520--529, 2011.

\bibitem[Bickel and Levina(2004)]{bickel2004some}
Peter~J Bickel and Elizaveta Levina.
\newblock Some theory for fisher's linear discriminant function,naive bayes',
  and some alternatives when there are many more variables than observations.
\newblock \emph{Bernoulli}, 10\penalty0 (6):\penalty0 989--1010, 2004.

\bibitem[Bose et~al.(2015)Bose, Pal, SahaRay, and Nayak]{bose2015generalized}
Smarajit Bose, Amita Pal, Rita SahaRay, and Jitadeepa Nayak.
\newblock Generalized quadratic discriminant analysis.
\newblock \emph{Pattern Recognition}, 48\penalty0 (8):\penalty0 2676--2684,
  2015.

\bibitem[Cai and Liu(2011)]{cai2011direct}
Tony Cai and Weidong Liu.
\newblock A direct estimation approach to sparse linear discriminant analysis.
\newblock \emph{Journal of the American Statistical Association}, 106\penalty0
  (496):\penalty0 1566--1577, 2011.

\bibitem[Cai et~al.(2011)Cai, Liu, and Luo]{cai2011constrained}
Tony Cai, Weidong Liu, and Xi~Luo.
\newblock A constrained l1 minimization approach to sparse precision matrix
  estimation.
\newblock \emph{Journal of the American Statistical Association}, 106\penalty0
  (494):\penalty0 594--607, 2011.

\bibitem[Candes and Tao(2007)]{candes2007dantzig}
Emmanuel Candes and Terence Tao.
\newblock The dantzig selector: Statistical estimation when p is much larger
  than n.
\newblock \emph{The Annals of Statistics}, 35\penalty0 (6):\penalty0
  2313--2351, 2007.

\bibitem[Cheng et~al.(2019)Cheng, Diakonikolas, and Ge]{cheng2019high}
Yu~Cheng, Ilias Diakonikolas, and Rong Ge.
\newblock High-dimensional robust mean estimation in nearly-linear time.
\newblock In \emph{Proceedings of the Thirtieth Annual ACM-SIAM Symposium on
  Discrete Algorithms}, pages 2755--2771. SIAM, 2019.

\bibitem[Clemmensen et~al.(2011)Clemmensen, Hastie, Witten, and
  Ersb{\o}ll]{clemmensen2011sparse}
Line Clemmensen, Trevor Hastie, Daniela Witten, and Bjarne Ersb{\o}ll.
\newblock Sparse discriminant analysis.
\newblock \emph{Technometrics}, 53\penalty0 (4):\penalty0 406--413, 2011.

\bibitem[Croux et~al.(2008)Croux, Filzmoser, and
  Joossens]{croux2008classification}
Christophe Croux, Peter Filzmoser, and Kristel Joossens.
\newblock Classification efficiencies for robust linear discriminant analysis.
\newblock \emph{Statistica Sinica}, pages 581--599, 2008.

\bibitem[Daradkeh et~al.(2021)Daradkeh, Gorokhovatskyi, Tvoroshenko, Gadetska,
  and Al-Dhaifallah]{daradkeh2021methods}
Yousef~Ibrahim Daradkeh, Volodymyr Gorokhovatskyi, Iryna Tvoroshenko, Svitlana
  Gadetska, and Mujahed Al-Dhaifallah.
\newblock Methods of classification of images on the basis of the values of
  statistical distributions for the composition of structural description
  components.
\newblock \emph{IEEE Access}, 9:\penalty0 92964--92973, 2021.

\bibitem[Deshmukh et~al.(2023)Deshmukh, Liu, and
  Veeravalli]{deshmukh2023robust}
Aditya Deshmukh, Jing Liu, and Venugopal~V Veeravalli.
\newblock Robust mean estimation in high dimensions: An outlier-fraction
  agnostic and efficient algorithm.
\newblock \emph{IEEE Transactions on Information Theory}, 69\penalty0
  (7):\penalty0 4675--4690, 2023.

\bibitem[Diakonikolas et~al.(2017)Diakonikolas, Kamath, Kane, Li, Moitra, and
  Stewart]{diakonikolas2017being}
Ilias Diakonikolas, Gautam Kamath, Daniel~M Kane, Jerry Li, Ankur Moitra, and
  Alistair Stewart.
\newblock Being robust (in high dimensions) can be practical.
\newblock In \emph{International Conference on Machine Learning}, pages
  999--1008. PMLR, 2017.

\bibitem[Diakonikolas et~al.(2019)Diakonikolas, Kamath, Kane, Li, Moitra, and
  Stewart]{diakonikolas2019robust}
Ilias Diakonikolas, Gautam Kamath, Daniel Kane, Jerry Li, Ankur Moitra, and
  Alistair Stewart.
\newblock Robust estimators in high-dimensions without the computational
  intractability.
\newblock \emph{SIAM Journal on Computing}, 48\penalty0 (2):\penalty0 742--864,
  2019.

\bibitem[Donoho et~al.(2005)Donoho, Elad, and Temlyakov]{donoho2005stable}
David~L Donoho, Michael Elad, and Vladimir~N Temlyakov.
\newblock Stable recovery of sparse overcomplete representations in the
  presence of noise.
\newblock \emph{IEEE Transactions on Information Theory}, 52\penalty0
  (1):\penalty0 6--18, 2005.

\bibitem[Fan and Fan(2008)]{fan2008high}
Jianqing Fan and Yingying Fan.
\newblock High dimensional classification using features annealed independence
  rules.
\newblock \emph{Annals of Statistics}, 36\penalty0 (6):\penalty0 2605, 2008.

\bibitem[Fang and Anderson(1990)]{fang1990statistical}
Kai-Tang Fang and Theodore~Wilbur Anderson.
\newblock Statistical inference in elliptically contoured and related
  distributions.
\newblock \emph{New York: Allerton Press Inc}, 1990.

\bibitem[Feng(2024)]{feng2024spatial}
Long Feng.
\newblock Spatial sign based principal component analysis for high dimensional
  data.
\newblock \emph{arXiv preprint arXiv:2409.13267}, 2024.

\bibitem[Feng and Liu(2017)]{feng2017high}
Long Feng and Binghui Liu.
\newblock High-dimensional rank tests for sphericity.
\newblock \emph{Journal of Multivariate Analysis}, 155:\penalty0 217--233,
  2017.

\bibitem[Feng and Sun(2016)]{feng2016spatial}
Long Feng and Fasheng Sun.
\newblock Spatial-sign based high-dimensional location test.
\newblock \emph{Electronic Journal of Statistics}, 10:\penalty0 2420--2434,
  2016.

\bibitem[Feng et~al.(2016)Feng, Zou, and Wang]{feng2016multivariate}
Long Feng, Changliang Zou, and Zhaojun Wang.
\newblock Multivariate-sign-based high-dimensional tests for the two-sample
  location problem.
\newblock \emph{Journal of the American Statistical Association}, 111\penalty0
  (514):\penalty0 721--735, 2016.

\bibitem[Guo et~al.(2007)Guo, Hastie, and Tibshirani]{guo2007regularized}
Yaqian Guo, Trevor Hastie, and Robert Tibshirani.
\newblock Regularized linear discriminant analysis and its application in
  microarrays.
\newblock \emph{Biostatistics}, 8\penalty0 (1):\penalty0 86--100, 2007.

\bibitem[Han et~al.(2013)Han, Zhao, and Liu]{han2013coda}
Fang Han, Tuo Zhao, and Han Liu.
\newblock {CODA}: High dimensional copula discriminant analysis.
\newblock \emph{Journal of Machine Learning Research}, 14\penalty0
  (Feb):\penalty0 629 -- 671, 2013.

\bibitem[Hastie(2009)]{hastie2009elements}
Trevor Hastie.
\newblock The elements of statistical learning: data mining, inference, and
  prediction, 2009.

\bibitem[Houdouin et~al.(2024)Houdouin, Jonckheere, and
  Pascal]{houdouin2024femda}
Pierre Houdouin, Matthieu Jonckheere, and Fr{\'e}d{\'e}ric Pascal.
\newblock Femda: a unified framework for discriminant analysis.
\newblock In \emph{Elliptically Symmetric Distributions in Signal Processing
  and Machine Learning}, pages 303--336. Springer, 2024.

\bibitem[Huang et~al.(2023)Huang, Liu, Zhou, and Feng]{huang2023high}
Xifen Huang, Binghui Liu, Qin Zhou, and Long Feng.
\newblock A high-dimensional inverse norm sign test for two-sample location
  problems.
\newblock \emph{Canadian Journal of Statistics}, 51\penalty0 (4):\penalty0
  1004--1033, 2023.

\bibitem[Le et~al.(2020)Le, Chaux, Richard, and Guedj]{le2020adapted}
Khuyen~T Le, Caroline Chaux, Fr{\'e}d{\'e}ric~JP Richard, and Eric Guedj.
\newblock An adapted linear discriminant analysis with variable selection for
  the classification in high-dimension, and an application to medical data.
\newblock \emph{Computational Statistics \& Data Analysis}, 152:\penalty0
  107031, 2020.

\bibitem[Ledoit and Wolf(2003)]{ledoit2003improved}
Olivier Ledoit and Michael Wolf.
\newblock Improved estimation of the covariance matrix of stock returns with an
  application to portfolio selection.
\newblock \emph{Journal of Empirical Finance}, 10\penalty0 (5):\penalty0
  603--621, 2003.

\bibitem[Li et~al.(2019)Li, Shao, Yin, and Liu]{li2019robust}
Chun-Na Li, Yuan-Hai Shao, Wotao Yin, and Ming-Zeng Liu.
\newblock Robust and sparse linear discriminant analysis via an alternating
  direction method of multipliers.
\newblock \emph{IEEE Transactions on Neural Networks and Learning Systems},
  31\penalty0 (3):\penalty0 915--926, 2019.

\bibitem[Li and Zhou(2017)]{li2017spectral}
Weiming Li and Wang Zhou.
\newblock On spectral properties of high-dimensional spatial-sign covariance
  matrices in elliptical distributions with applications.
\newblock \emph{arXiv preprint arXiv:1705.06427}, 2017.

\bibitem[Lu and Feng(2025)]{lu2025robust}
Zhengke Lu and Long Feng.
\newblock Robust sparse precision matrix estimation and its application.
\newblock \emph{arXiv preprint arXiv:2503.03575}, 2025.

\bibitem[Mai et~al.(2012)Mai, Zou, and Yuan]{mai2012direct}
Qing Mai, Hui Zou, and Ming Yuan.
\newblock A direct approach to sparse discriminant analysis in ultra-high
  dimensions.
\newblock \emph{Biometrika}, 99\penalty0 (1):\penalty0 29--42, 2012.

\bibitem[Oja(2010)]{oja2010multivariate}
Hannu Oja.
\newblock \emph{Multivariate nonparametric methods with R: an approach based on
  spatial signs and ranks}.
\newblock Springer Science and Business Media, 2010.

\bibitem[Park et~al.(2022)Park, Baek, and Park]{park2022high}
Hoyoung Park, Seungchul Baek, and Junyong Park.
\newblock High-dimensional linear discriminant analysis using nonparametric
  methods.
\newblock \emph{Journal of Multivariate Analysis}, 188:\penalty0 104836, 2022.

\bibitem[Raninen et~al.(2021)Raninen, Tyler, and Ollila]{raninen2021linear}
Elias Raninen, David~E Tyler, and Esa Ollila.
\newblock Linear pooling of sample covariance matrices.
\newblock \emph{IEEE Transactions on Signal Processing}, 70:\penalty0 659--672,
  2021.

\bibitem[Rublee et~al.(2011)Rublee, Rabaud, Konolige, and
  Bradski]{rublee2011orb}
Ethan Rublee, Vincent Rabaud, Kurt Konolige, and Gary Bradski.
\newblock Orb: An efficient alternative to sift or surf.
\newblock In \emph{2011 International Conference on Computer Vision}, pages
  2564--2571. IEEE, 2011.

\bibitem[Sch{\"a}fer and Strimmer(2005)]{schafer2005shrinkage}
Juliane Sch{\"a}fer and Korbinian Strimmer.
\newblock A shrinkage approach to large-scale covariance matrix estimation and
  implications for functional genomics.
\newblock \emph{Statistical Applications in Genetics and Molecular Biology},
  4\penalty0 (1), 2005.

\bibitem[Shao et~al.(2011)Shao, Wang, Deng, and Wang]{shao2011sparse}
Jun Shao, Yazhen Wang, Xinwei Deng, and Sijian Wang.
\newblock Sparse linear discriminant analysis by thresholding for high
  dimensional data.
\newblock \emph{The Annals of Statistics}, 39\penalty0 (2):\penalty0
  1241--1265, 2011.

\bibitem[Shi et~al.(2009)Shi, Dai, Liu, and Yan]{shi2009sparse}
Yu~Shi, Daoqing Dai, Chaochun Liu, and Hong Yan.
\newblock Sparse discriminant analysis for breast cancer biomarker
  identification and classification.
\newblock \emph{Progress in Natural Science}, 19\penalty0 (11):\penalty0
  1635--1641, 2009.

\bibitem[Tibshirani et~al.(2002)Tibshirani, Hastie, Narasimhan, and
  Chu]{tibshirani2002diagnosis}
Robert Tibshirani, Trevor Hastie, Balasubramanian Narasimhan, and Gilbert Chu.
\newblock Diagnosis of multiple cancer types by shrunken centroids of gene
  expression.
\newblock \emph{Proceedings of the National Academy of Sciences}, 99\penalty0
  (10):\penalty0 6567--6572, 2002.

\bibitem[Wakaki(1994)]{wakaki1994discriminant}
Hirofumi Wakaki.
\newblock Discriminant analysis under elliptical populations.
\newblock \emph{Hiroshima Mathematical Journal}, 24\penalty0 (2):\penalty0
  257--298, 1994.

\bibitem[Wang et~al.(2015)Wang, Peng, and Li]{wang2015high}
Lan Wang, Bo~Peng, and Runze Li.
\newblock A high-dimensional nonparametric multivariate test for mean vector.
\newblock \emph{Journal of the American Statistical Association}, 110\penalty0
  (512):\penalty0 1658--1669, 2015.

\bibitem[Witten and Tibshirani(2009)]{witten2009covariance}
Daniela~M Witten and Robert Tibshirani.
\newblock Covariance-regularized regression and classification for high
  dimensional problems.
\newblock \emph{Journal of the Royal Statistical Society Series B: Statistical
  Methodology}, 71\penalty0 (3):\penalty0 615--636, 2009.

\bibitem[Yang et~al.(2023)Yang, Lin, and Li]{yang2023efficient}
Hannan Yang, DY~Lin, and Quefeng Li.
\newblock An efficient greedy search algorithm for high-dimensional linear
  discriminant analysis.
\newblock \emph{Statistica Sinica}, 33\penalty0 (SI):\penalty0 1343, 2023.

\bibitem[Zhang et~al.(2022)Zhang, Zhao, and Feng]{zhang2022robust}
Xiaoxu Zhang, Ping Zhao, and Long Feng.
\newblock Robust sphericity test in the panel data model.
\newblock \emph{Statistics and Probability Letters}, 182:\penalty0 109304,
  2022.

\bibitem[Zhu et~al.(2022)Zhu, Jiao, and Steinhardt]{zhu2022robust}
Banghua Zhu, Jiantao Jiao, and Jacob Steinhardt.
\newblock Robust estimation via generalized quasi-gradients.
\newblock \emph{Information and Inference: A Journal of the IMA}, 11\penalty0
  (2):\penalty0 581--636, 2022.

\bibitem[Zou et~al.(2014)Zou, Peng, Feng, and Wang]{zou2014multivariate}
Changliang Zou, Liuhua Peng, Long Feng, and Zhaojun Wang.
\newblock Multivariate sign-based high-dimensional tests for sphericity.
\newblock \emph{Biometrika}, 101\penalty0 (1):\penalty0 229--236, 2014.

\end{thebibliography}

}

\newpage

{\Large Supplementray Material of "Spatial Sign based Direct Sparse Linear Discriminant Analysis for High Dimensional Data"}
\appendix
\section{Lemma}

\begin{lemma}
Under (C1) and (C2), we have with probability greater than $1-O(p^{-1})$,
\begin{align}\label{lm1}
\|p\hat{\S}\L^{-1}\bm \delta-(\tilde{\bmu}_1-\tilde{\bmu}_2)\|_\infty\le \lambda_n
\end{align}
\end{lemma}
\begin{proof}
By the triangle inequality, we have
\begin{align*}
&\|p\hat{\S}\L^{-1}\bm \delta-(\tilde{\bmu}_1-\tilde{\bmu}_2)\|_\infty\\
&\le \|p\hat{\S}\L^{-1}\bm \delta-p{\S}\L^{-1}\bm \delta\|_\infty+\|p{\S}\L^{-1}\bm \delta-\bm \delta\|_\infty+\|\tilde{\bmu}_1-\bmu_1\|_\infty+\|\tilde{\bmu}_2-\bmu_2\|_\infty
\end{align*}
According to the proof of Theorem 7.1 in 
\cite{feng2024spatial}, we have
\begin{align*}
\|\tilde{\bmu}_k-\bmu_k\|_\infty \le C\sqrt{\log p/n}, k=1,2, \|p\hat{\S}-p\S\|_\infty\le C\sqrt{\log p/n}
\end{align*}
for some large enough constant $C>0$ with probability larger than $1-O(p^{-1})$. So
\begin{align*}
\|p\hat{\S}\L^{-1}\bm \delta-p{\S}\L^{-1}\bm \delta\|_\infty \le \|\L^{-1}\bm \delta\|_1 \|p\hat{\S}-p\S\|_\infty \le \lambda_n/4.
\end{align*}
Additionally,
\begin{align*}
\|p{\S}\L^{-1}\bm \delta-\bm \delta\|_\infty=\|(p{\S}-\L)\L^{-1}\bm \delta\|_\infty\le \|\L^{-1}\bm \delta\|_1 \|p{\S}-\L\|_\infty \le \lambda_n/4
\end{align*}
by the assumption.
\end{proof}

\section{Proof of Theorems} \label{App: prof}

\subsection{Proof of Theorem 3.1} \label{pfth1}
\begin{proof}
	
	By the definition of $\hat{\boldsymbol{\gamma}}$, we have
\begin{align}\label{the1eq1}
    \left|p(\L^{-1} \boldsymbol{\delta})^{\prime}\hat{\S} \hat{\boldsymbol{\gamma}}-(\L^{-1} \boldsymbol{\delta})^{\prime}(\tilde{\bmu}_1-\tilde{\bmu}_2)\right| \leq& ~\lambda_n \left\| \L^{-1}\boldsymbol{\delta} \right\|_1 + \|\tilde{\boldsymbol{\delta}}-\boldsymbol{\delta} \|_{\infty} \left\| \L^{-1}\boldsymbol{\delta} \right\|_1  \nonumber\\
      \leq & ~2\lambda_n \left\| \L^{-1}\boldsymbol{\delta} \right\|_1. 
\end{align} 

By (\ref{lm1}), we have
\begin{align}\label{the1eq2}
\left|p(\L^{-1} \boldsymbol{\delta})^{\prime}\hat{\S} \hat{\boldsymbol{\gamma}}-\boldsymbol{\delta}^{\prime} \hat{\boldsymbol{\gamma}} \right| \leq \lambda_n \|\hat{\boldsymbol{\gamma}}\|_1+ \|\tilde{\boldsymbol{\delta}}-\boldsymbol{\delta}\|_{\infty}\|\hat{\boldsymbol{\gamma}}\|_1 \leq 2\lambda_n \left\| \L^{-1}\boldsymbol{\delta} \right\|_1,
\end{align}
and together with (\ref{the1eq1}) implies 
\begin{align}\label{the1eq3}
	|(\hat{\boldsymbol{\gamma}}-\L^{-1}\boldsymbol{\delta})^{\prime}  \boldsymbol{\delta} | \leq 4\lambda_n \left\| \L^{-1}\boldsymbol{\delta} \right\|_1.	
\end{align}	

Then we have 
\begin{align}\label{the1eq4}
	|(\tilde{\bmu}-\bmu_1)^{\prime}\hat{\boldsymbol{\gamma}} +\frac{1}{2} \boldsymbol{\delta}^{\prime} \L^{-1} \boldsymbol{\delta}| \leq & ~|(\tilde{\bmu}-\bmu)^{\prime}\hat{\boldsymbol{\gamma}}|+ \frac{1}{2}|\boldsymbol{\delta}^{\prime}\hat{\boldsymbol{\gamma}} - \boldsymbol{\delta}^{\prime} \L^{-1} \boldsymbol{\delta}| \nonumber\\
	\leq& ~|(\tilde{\bmu}-\bmu)^{\prime}\hat{\boldsymbol{\gamma}}| +2\lambda_n \left\| \L^{-1}\boldsymbol{\delta} \right\|_1 \nonumber\\
	\leq&
	~C\sqrt{\frac{\log p}{n}}\left\| \L^{-1}\boldsymbol{\delta} \right\|_1+2\lambda_n \left\| \L^{-1}\boldsymbol{\delta} \right\|_1. 
\end{align} 
Similarly, we have
\begin{align}\label{the1eq5}
	|(\tilde{\bmu}-\bmu_2)^{\prime}\hat{\boldsymbol{\gamma}} -\frac{1}{2} \boldsymbol{\delta}^{\prime} \L^{-1} \boldsymbol{\delta}| 
	\leq
	C\sqrt{\frac{\log p}{n}}\left\| \L^{-1}\boldsymbol{\delta} \right\|_1+2\lambda_n \left\| \L^{-1}\boldsymbol{\delta} \right\|_1. 
\end{align}
Next, considering the denominator in $R_n$. We have
\begin{align*}
    \|\boldsymbol{\Sigma} \hat{\boldsymbol{\gamma}}-\boldsymbol{\delta}\|_{\infty} \leq \|\boldsymbol{\Sigma} \hat{\boldsymbol{\gamma}}-p\hat{\S} \hat{\boldsymbol{\gamma}}\|_{\infty}+2\lambda_n
    \leq C \left\| \L^{-1}\boldsymbol{\delta} \right\|_1 \sqrt{\frac{\log p}{n}} + 2\lambda_n.
\end{align*}
Thus we have
\begin{align*}
	|\hat{\boldsymbol{\gamma}}^{\prime}\boldsymbol{\Sigma} \hat{\boldsymbol{\gamma}}-\hat{\boldsymbol{\gamma}}^{\prime}\boldsymbol{\delta}|
	\leq C \left\| \L^{-1}\boldsymbol{\delta} \right\|_1^2 \sqrt{\frac{\log p}{n}} + 2\lambda_n\left\| \L^{-1}\boldsymbol{\delta} \right\|_1.
\end{align*}
According to (\ref{the1eq3}), we have
\begin{align}\label{the1eq6}
	|\hat{\boldsymbol{\gamma}}^{\prime}\boldsymbol{\Sigma} \hat{\boldsymbol{\gamma}}-\boldsymbol{\delta}^{\prime}\L^{-1}\boldsymbol{\delta}| 
	\leq C \left\| \L^{-1}\boldsymbol{\delta} \right\|_1^2 \sqrt{\frac{\log p}{n}} + 6\lambda_n\left\| \L^{-1}\boldsymbol{\delta} \right\|_1.
\end{align}
Suppose $\boldsymbol{\delta}^{\prime} \L^{-1}\boldsymbol{\delta}\geq M$ for some $M>0$. By (\ref{the1}), (\ref{the1eq4}) and (\ref{the1eq6}), we have
\begin{align*}
\left|	\frac{(\tilde{\bmu}-\bmu_1)^{\prime}\hat{\boldsymbol{\gamma}}}{\sqrt{\hat{\boldsymbol{\gamma}}^{\prime}\boldsymbol{\Sigma} \hat{\boldsymbol{\gamma}}}} \right| \geq C \left|	\frac{\boldsymbol{\delta}^{\prime}\L^{-1}\boldsymbol{\delta}}{\sqrt{\hat{\boldsymbol{\gamma}}^{\prime}\boldsymbol{\Sigma} \hat{\boldsymbol{\gamma}}}} \right| \geq C((\boldsymbol{\delta}^{\prime}\L^{-1}\boldsymbol{\delta})^{-1}+o(1))^{-1/2}\geq CM^{1/2},
\end{align*}
this implies that
\begin{align}\label{the1eq7}
	|R_n-R |\leq \text{exp}(-CM).
\end{align}
Suppose $\boldsymbol{\delta}^{\prime} \L^{-1}\boldsymbol{\delta}\leq M$, by (\ref{the1}), and (\ref{the1eq6}), yields 
\begin{align}\label{the1eq8}
	\left|\frac{\hat{\boldsymbol{\gamma}}^{\prime}\boldsymbol{\Sigma} \hat{\boldsymbol{\gamma}}}{\boldsymbol{\delta}^{\prime}\L^{-1}\boldsymbol{\delta}}-1 \right|=o(1).
\end{align}
And together with (\ref{the1eq4}), we have 
\begin{align}\label{the1eq9}
	\left|\frac{(\tilde{\bmu}-\bmu_1)^{\prime}\hat{\boldsymbol{\gamma}}}{\sqrt{\hat{\boldsymbol{\gamma}}^{\prime}\boldsymbol{\Sigma} \hat{\boldsymbol{\gamma}}}}+\frac{(1/2)\boldsymbol{\delta}^{\prime}\L^{-1}\boldsymbol{\delta}}{\sqrt{\hat{\boldsymbol{\gamma}}^{\prime}\boldsymbol{\Sigma} \hat{\boldsymbol{\gamma}}}}	\right|\leq C\frac{|\L^{-1} \boldsymbol{\delta}|}{(\boldsymbol{\delta}^{\prime}\L^{-1}\boldsymbol{\delta})^{1/2}}\lambda_n.
\end{align}
By (\ref{the1eq6}), we have
\begin{align}\label{the1eq10}
	\left|\frac{1}{\sqrt{\hat{\boldsymbol{\gamma}}^{\prime}\boldsymbol{\Sigma} \hat{\boldsymbol{\gamma}}}}-\frac{1}{\sqrt{\boldsymbol{\delta}^{\prime}\L^{-1}\boldsymbol{\delta}}}\right|
     \leq&  \frac{C\|\L^{-1}\boldsymbol{\delta}\|_1^2\sqrt{\log p/n}+6\|\L^{-1}\boldsymbol{\delta}\|_1\lambda_n}{\sqrt{\hat{\boldsymbol{\gamma}}^{\prime}\boldsymbol{\Sigma} \hat{\boldsymbol{\gamma}}}\sqrt{\boldsymbol{\delta}^{\prime}\L^{-1}\boldsymbol{\delta}}\left(\sqrt{\hat{\boldsymbol{\gamma}}^{\prime}\boldsymbol{\Sigma} \hat{\boldsymbol{\gamma}}}+\sqrt{\boldsymbol{\delta}^{\prime}\L^{-1}\boldsymbol{\delta}} \right)}\nonumber\\
	 \leq& C (\boldsymbol{\delta}^{\prime}\L^{-1}\boldsymbol{\delta})^{-3/2}\left(\|\L^{-1}\boldsymbol{\delta}\|_1^2\sqrt{\frac{\log p}{n}}+ \|\L^{-1}\boldsymbol{\delta}\|_1 \lambda_n \right) .
\end{align}
and
\begin{align}\label{the1eq11}
	\left| \frac{(1/2)\boldsymbol{\delta}^{\prime}\L^{-1}\boldsymbol{\delta}}{\sqrt{\hat{\boldsymbol{\gamma}}^{\prime}\boldsymbol{\Sigma} \hat{\boldsymbol{\gamma}}}} -\frac{1}{2}(\boldsymbol{\delta}^{\prime}\L^{-1}\boldsymbol{\delta})^{1/2}\right|
	\leq 
	C \frac{\|\L^{-1}\boldsymbol{\delta}\|_1^2}{(\boldsymbol{\delta}^{\prime}\L^{-1}\boldsymbol{\delta})^{1/2}}\sqrt{\frac{\log p}{n}}+C\frac{\|\L^{-1}\boldsymbol{\delta}\|_1}{(\boldsymbol{\delta}^{\prime}\L^{-1}\boldsymbol{\delta})^{1/2}}\lambda_n=:r_n.
\end{align}
By condition 1, (\ref{the1eq9}) and (\ref{the1eq11}),
\begin{align}\label{the1eq12}
R_n= R\times (1+O(1)r_n(\boldsymbol{\delta}^{\prime}\L^{-1}\boldsymbol{\delta})^{1/2}\exp(O(1)(\boldsymbol{\delta}^{\prime}\L^{-1}\boldsymbol{\delta})^{1/2}r_n)).
\end{align}
By the assumption $\boldsymbol{\delta}^{\prime} \L^{-1}\boldsymbol{\delta}\leq M$ and the condition (\ref{the1}), we have $(\|\boldsymbol{\Omega} \boldsymbol{\delta}\|_1+\|\boldsymbol{\Omega} \boldsymbol{\delta}\|_1^2)\sqrt{\log p/n}=o(1)$, thus $R_n=(1+o(1))R$, when letting $n,p\rightarrow\infty$ first and then $M\rightarrow\infty$.
The proof is completed.
\end{proof}

\subsection{Proof of Theorem 3.2} \label{pfth2}

\begin{proof}
According to Lemma 6 in \cite{lu2025robust}, we know $\|p\mathbf{S}-\mathbf{\Sigma}\|_\infty=O(p^{-1/2})$. Additionally, by Lemma 7 in \cite{lu2025robust}, we have $\|p\hat{\mathbf{S}}-p\mathbf{S}\|_\infty=O_p(\sqrt{\log p/n})$. Thus, $\|p\hat{\mathbf{S}}-\mathbf{\Sigma}\|_\infty=O_p(\sqrt{\log p/n})$ if $\frac{n}{p\log p }\to 0$.
Since 
\begin{align}\label{the1eq13}
	\|\boldsymbol{\Sigma}(\hat{\boldsymbol{\gamma}}-\L^{-1}\boldsymbol{\delta})\|_{\infty} \leq & \| p\hat{\S}(\hat{\boldsymbol{\gamma}}-\L^{-1}\boldsymbol{\delta})\|_{\infty} + \|(p\hat{\S}-\boldsymbol{\Sigma})(\hat{\boldsymbol{\gamma}}-\L^{-1}\boldsymbol{\delta})\|_{\infty} \nonumber\\
	\leq & 2\lambda_n + C\|\hat{\boldsymbol{\gamma}}-\L^{-1}\boldsymbol{\delta}\|_{1}\sqrt{\frac{\log p}{n}} \nonumber\\
	\leq & 2 \lambda_n +C \|\L^{-1}\boldsymbol{\delta}\|_0 \sqrt{\frac{\log p}{n}}\|\hat{\boldsymbol{\gamma}}-\L^{-1}\boldsymbol{\delta}\|_{\infty}  \nonumber\\
	\leq & 2 \lambda_n+C \Vert \L^{-1}\Vert_{L_1} \|\L^{-1}\boldsymbol{\delta}\|_0 \sqrt{\frac{\log p}{n}}\|\boldsymbol{\Sigma}(\hat{\boldsymbol{\gamma}}-\L^{-1}\boldsymbol{\delta})\|_{\infty},
\end{align}
together with $\Vert \L^{-1}\Vert_{L_1} \|\L^{-1}\boldsymbol{\delta}\|_0 \sqrt{\frac{\log p}{n}}=o(1)$ implies that $\|\boldsymbol{\Sigma}(\hat{\boldsymbol{\gamma}}-\boldsymbol{\gamma})\|_{\infty}\leq C\lambda_n$, we have
\begin{align*}
	|\hat{\boldsymbol{\gamma}}^{\prime}\boldsymbol{\Sigma} \hat{\boldsymbol{\gamma}}-\hat{\boldsymbol{\gamma}}^{\prime}\boldsymbol{\Sigma}\L^{-1}\boldsymbol{\delta}|\leq C\| \L^{-1}\boldsymbol{\delta}\|_1\lambda_n,
\end{align*}
and 
\begin{align*}
|\hat{\boldsymbol{\gamma}}^{\prime}\boldsymbol{\Sigma}\L^{-1}\boldsymbol{\delta}-\boldsymbol{\delta}^{\prime}\L^{-1}\boldsymbol{\delta}|\leq C\| \L^{-1}\boldsymbol{\delta}\|_1\lambda_n.
\end{align*}
The remaining steps refer to the proof of (\ref{the1eq12}),
and the proof is completed.
\end{proof}

\end{document}